\documentclass[11pt,letter]{article}
\usepackage{authblk}
\usepackage[margin=1in]{geometry}
\usepackage[utf8x]{inputenc}
\usepackage{amsthm}
\usepackage{wrapfig,floatflt,graphicx,amssymb,textcomp,array,amsmath}
\usepackage{enumerate,enumitem}
\usepackage{multirow}
\usepackage{tabularx}
\usepackage{color}
\usepackage{todonotes}
\usepackage[titletoc,title]{appendix}

\usepackage{lineno}

\newcommand{\etal}{{et~al.}}

\title{Acute Tours in the Plane
}

\author{Ahmad Biniaz\thanks{This research is supported by NSERC.}
}

\affil{School of Computer Science, University of Windsor\\\texttt{ahmad.biniaz@gmail.com}}

\date{}
\newtheorem{lemma}{Lemma}

\newtheorem{theorem}{Theorem}
\newtheorem{observation}{Observation}
\newtheorem*{problem*}{Problem}
\newtheorem*{claim*}{Claim}
\newtheorem*{invariant*}{Invariant}

\renewcommand\footnotemark{}
\begin{document}
	\maketitle
	\vspace{-10pt}
	\begin{abstract}
		We confirm the following conjecture of Fekete and Woeginger from 1997: for any sufficiently large even
		number $n$, every set of $n$ points in the plane can be connected by a spanning tour (Hamiltonian cycle) consisting of straight-line edges such that the angle between any two consecutive edges is at most $\pi/2$. Our proof is constructive and suggests a simple $O(n\log n)$-time algorithm for finding such a tour. The previous best-known upper bound on the angle is $2\pi/3$, and it is due to Dumitrescu, Pach and T{\'{o}}th (2009). 
	\end{abstract}

\section{Introduction}

The Euclidean traveling salesperson problem (TSP) is a well-studied and fundamental problem in combinatorial optimization and computational geometry. In this problem we are given a set of points in the plane and our goal is to find a shortest tour that visits all points. Motivated by applications in robotics and motion planning, in recent years there has been an increased interest in the study of tours with bounded angles at vertices, rather than bounded length of edges; see e.g. \cite{Aggarwal1999,Aichholzer2017,Dumitrescu12,Fekete1992,Fekete1997} and references therein.  Bounded-angle structures (tours, paths, trees) are also desirable in the context of designing networks with directional antennas \cite{Aschner2017,Aschner2012,Carmi2011,Tran2017}. Bounded-angle tours (and paths), in particular, have received considerable attention following the PhD thesis of S. Fekete \cite{Fekete1992} and the seminal work of Fekete and Woeginger \cite{Fekete1997}.

Consider a set $P$ of at least three points in the plane. A {\em spanning tour} is a directed Hamiltonian cycle on $P$ that is drawn with straight-line edges. When three consecutive vertices $p_i, p_{i+1},p_{i+2}$ of the tour are traversed in this order, the {\em rotation angle} at $p_{i+1}$ (denoted by $\angle p_ip_{i+1}p_{i+2}$) is the angle in $[0,\pi]$ that is determined by the segments $p_ip_{i+1}$ and $p_{i+1}p_{i+2}$. 
If all rotation angles in a tour are at most $\pi/2$ then it is called an {\em acute} tour.

In 1997, Fekete and Woeginger \cite{Fekete1997} raised many challenging questions about bounded-angle tours and paths. In particular they conjectured that {\em for any sufficiently large even number $n$, every set of $n$ points in the plane admits an acute spanning tour $($a tour with rotation angles at most $\pi/2$$)$}. They stated the conjecture specifically for $n\geqslant 8$. The point set illustrated in Figure~\ref{lower-bound-fig}(a) (also described in \cite{Fekete1997}) shows that the upper bound $\pi/2$ is the best achievable. The conjecture does not hold if $n$ is allowed to be an odd number; for example if the $n$ points are on a line then in any spanning tour one of the rotation angles must be $\pi$. The conjecture also does not hold if $n$ is allowed to be small. For instance the 4-element point set consisting of the 3 vertices of an equilateral triangle with its center, must have a rotation angle $2\pi/3$ in any spanning tour. Also the 6-element point set of Figure~\ref{lower-bound-fig}(b) (also illustrated in \cite{Fekete1997} and \cite{Dumitrescu12}) must have a rotation angle of at least $2\pi/3-\epsilon$ in any spanning tour, for some arbitrary small constant $\epsilon$.

\begin{figure}[htb]
	\centering
	\setlength{\tabcolsep}{0in}
	$\begin{tabular}{cc}
		\multicolumn{1}{m{.47\columnwidth}}{\centering\includegraphics[width=.26\columnwidth]{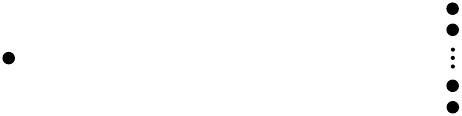}}
		&\multicolumn{1}{m{.53\columnwidth}}{\centering\vspace{0pt}\includegraphics[width=.38\columnwidth]{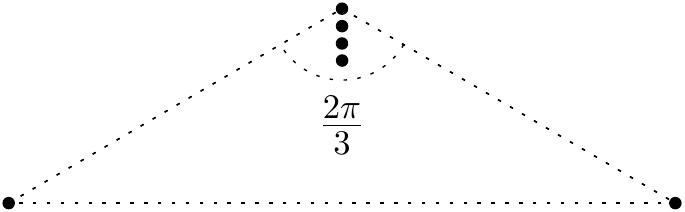}}	
		\\
		(a) &(b)  
	\end{tabular}$
	\caption{(a) a general lower bound example, and (b) a lower bound example for $6$ points.}
	\label{lower-bound-fig}
\end{figure}

In 2009, Dumitrescu, Pach and T{\'{o}}th \cite{Dumitrescu12} took the first promising steps towards proving the conjecture. They confirmed the conjecture for points in convex position. For general point sets, they obtained the first partial result by showing that any point set (with even number of points) admits a spanning tour in which each rotation angle is at most $2\pi/3$. 

In this paper we prove the conjecture of Fekete and Woeginger for general point sets.

\begin{theorem}
	\label{main-thr}
	Let $n\geqslant 20$ be an even integer. Then every set of $n$ points in the plane admits an acute spanning tour. Such a tour can be computed in linear time after finding an equitable partitioning of points with two orthogonal lines.
\end{theorem}

Due to our desire of having a short proof, we prove the conjecture for $n\geqslant 20$. Perhaps with some detailed case analysis one could extend the range of $n$ to a number smaller than $20$.
\paragraph{Difficulties towards a proof.} Fekete and Woeginger \cite{Fekete1997} exhibited an arbitrary-large even-size point set for which an algorithm (or a proof technique), that always outputs the longest tour or includes the diameter in the solution, does not achieve an acute tour; the point set is similar to that of Figure~\ref{lower-bound-fig}(b) but has more than $6$ points. This somehow breaks the hope for finding an acute tour by using greedy techniques. Therefore, to prove the conjecture one might need to employ some nontrivial ideas.  

\subsection{Related problems}

Another interesting conjecture of Fekete and Woeginger \cite{Fekete1997} is that any set of points in the plane admits a spanning path in which all rotation angles are at least $\pi/6$.\footnote{This bound is the best achievable as the three vertices of an equilateral triangle together with its center do not admit a path with rotation angles greater than $\pi/6$.} In 2008, B{\'{a}}r{\'{a}}ny, P{\'{o}}r, and
Valtr \cite{Barany2009} obtained the first constant lower bound of $\pi/9$, thereby gave a partial answer to the conjecture. The full conjecture was then proved, although not yet written in a paper format, by J. Kyn\v{c}l \cite{Kyncl2019} (see also the note added in the proof of \cite{Barany2009}).

Fekete and Woeginger \cite{Fekete1997} showed that any set of points in the plane admits an acute spanning path (where all intermediate rotation angles are at most $\pi/2$). Such a path can be obtained simply by starting from an arbitrary point and iteratively connecting the current point to its farthest among the remaining points. 
Notice that the resulting path always contains the diameter and by the difficulties mentioned above it cannot be completed to an acute tour. 
Carmi \etal~\cite{Carmi2011} showed how to construct acute paths with shorter edges; again no guarantee to be completed to an acute tour. 
Aichholzer \etal~\cite{Aichholzer2013} studied a similar problem with an additional constraint that the path should be {\em plane} (i.e., its edges do not cross each other). Among other results, they showed that any set of points in the plane in general position admits a plane spanning path with rotation angles at most $3\pi/4$. They also conjectured that this upper bound could be replaced by $\pi/2$.

The bounded-angle minimum spanning tree (also known as $\alpha$-MST) is a related problem that asks for a Euclidean minimum spanning tree in which all edges incident to every vertex lie in a cone of angle at most $\alpha$. This problem is motivated by replacing omni-directional antennas---in a wireless network---with directional antennas which are more secure, require lower transmission ranges, and cause less interference; see e.g. \cite{Aschner2017,Aschner2012,Biniaz2020,Biniaz2022,Tran2017}.

Another related problem (with an  objective somewhat opposite to ours) is to minimize the total {\em turning angle} of the tour \cite{Aggarwal1999}.\footnote{The turning angle at a vertex $v$ is the change in the direction of motion at $v$ when traveling on the tour. It is essentially $\pi$ minus the rotation angle at $v$.} Similar problems also studied under {\em pseudo-convex} tours and paths (that make only right turns) \cite{Fekete1997} and {\em reflexivity} of a point set (the smallest number
of reflex vertices in a simple polygonalization of the point set) \cite{Ackerman2009,Arkin2003}.

The so-called {\em Tverberg cycle} is a cycle with straight-line edges such that the diametral disks\footnote{The diametral disk induced by an edge $pq$ is the disk that has $pq$ as its diameter.} induced by the edges have nonempty intersection. Recently, Pirahmad \etal~\cite{Pirahmad2021} showed how to construct a spanning Tverberg cycle on any set of points in the plane. Although the constructed cycle has many acute angles, it is still far from being fully acute. 

\paragraph{Remark.}It is worth mentioning that having a tour with many acute angles, does not necessarily help in getting a fully acute tour because one can simply get a tour with at least $n-2$ acute angles by interconnecting the endpoints of acute paths obtained in \cite{Carmi2011,Fekete1997}.

\section{Preliminaries for the proof}
\label{preliminaries}

A set of four points in the plane is called a {\em quadruple}. If the four points are in convex position then the quadruple is called {\em convex}, otherwise it is called {\em concave}; the quadruple in Figure~\ref{hook-fig}(a) is convex while the quadruples in Figures~\ref{hook-fig}(b) and \ref{hook-fig}(c) are concave. We refer to the interior point of a concave quadruple as its {\em center}. By connecting the center of a concave quadruple to its other three points we obtain three angles. If one of these angles is at most $\pi/2$ then the quadruple is called {\em concave-acute}, otherwise all the angles are larger than $\pi/2$ and the quadruple is called {\em concave-obtuse}; the quadruple in Figure~\ref{hook-fig}(b) is concave-acute while the one in Figure~\ref{hook-fig}(c) is concave-obtuse.

A path, that is drawn by straight-line edges, is called {\em acute} if all the angles determined by its adjacent edges are at most $\pi/2$.
For two directed paths $P_1$ and $P_2$, where $P_1$ ends at the same vertex at which $P_2$ starts, we denote their concatenation by $P_1\oplus P_2$. 

For two distinct points $p$ and $q$ in the plane, we say that $p$ is {\em to the left of} $q$ if the $x$-coordinate of $p$ is not larger than the $x$-coordinate of $q$.
Analogously, we say that $p$ is {\em below} $q$ if the $y$-coordinate of $p$ is not larger than the $y$-coordinate of $q$.

It is known that any set of $n$ points in the plane can be split into four parts of equal size using two orthogonal lines (see e.g. \cite{Roy2007} or \cite[Section 6.6]{Courant1979}); such two lines can be computed in $\Theta(n\log n)$ time \cite{Roy2007}. The following is a restatement of this result which is borrowed from \cite{Dumitrescu12}.

\begin{lemma}
	\label{equitable-partition-lemma}
	Given a set $S$ of $n$ points in the plane $(n$ even$)$, one can always find two orthogonal lines $\ell_1$, $\ell_2$ and a partition $S = S_1 \cup S_2 \cup S_3 \cup S_4$ with $|S_1| = |S_3| = \lfloor \frac{n}{4}\rfloor$ and $|S_2| = |S4| = \lceil \frac{n}{4}\rceil$ such that $S_1$ and $S_3$ belong to two opposite closed quadrants determined by $\ell_1$ and $\ell_2$, and $S_2$ and $S_4$ belong to the other two opposite closed quadrants.
\end{lemma}

Our proof of Theorem~\ref{main-thr} shares some similarities with that of Dumitrescu \etal~\cite{Dumitrescu12} (for points in convex position) in the sense that both proofs employ the equitable partitioning of Lemma~\ref{equitable-partition-lemma}. However, there are major differences between the two proofs mainly because simple structures, that appear in points in convex position, do not necessarily appear in general point sets. Therefore one needs to extract complex structures from general point sets and combine them to establish a proof.


\section{Proof of Theorem~\ref{main-thr}}

Throughout this section we assume that $n$ is an even integer. We show how to construct an acute tour on any set of $n\geqslant 20$ points in the plane, and thus proving Theorem~\ref{main-thr}. In Subsection~\ref{setup-section} we describe the setup for our construction, and then in Subsection~\ref{tour-section} we construct the tour.  

\subsection{The proof setup}

\label{setup-section}

Let $S$ be a set of $n\geqslant 20$ points in the plane. Let $\{S_1, S_2, S_3, S_4\}$ be an equitable partitioning of $S$ with two orthogonal lines $\ell_1$ and $\ell_2$ that satisfies the conditions of Lemma~\ref{equitable-partition-lemma}. After a suitable rotation and translation we may assume that $\ell_1$ and $\ell_2$ coincide with the $x$ and $y$ coordinate axes, respectively. Also, after a suitable relabeling we may assume that all points of $S_i$ belong the $i$th quadrant determined by the axes as depicted in Figure~\ref{hook-fig}(a).

\begin{figure}[htb]
	\centering
	\setlength{\tabcolsep}{0in}
	$\begin{tabular}{ccc}
		\multicolumn{1}{m{.34\columnwidth}}{\centering\includegraphics[width=.28\columnwidth]{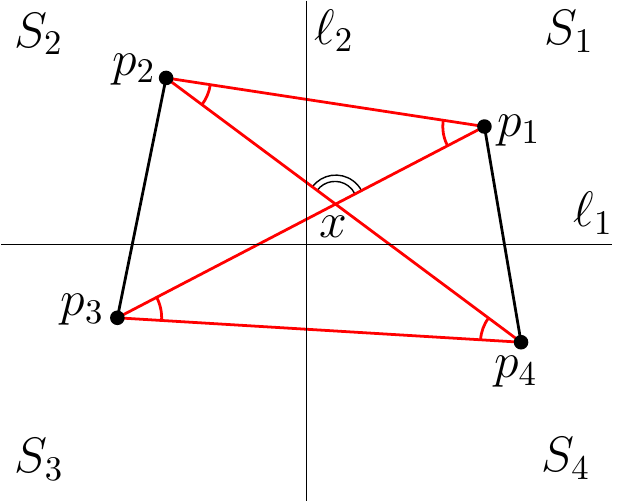}}
		&\multicolumn{1}{m{.33\columnwidth}}{\centering\vspace{0pt}\includegraphics[width=.28\columnwidth]{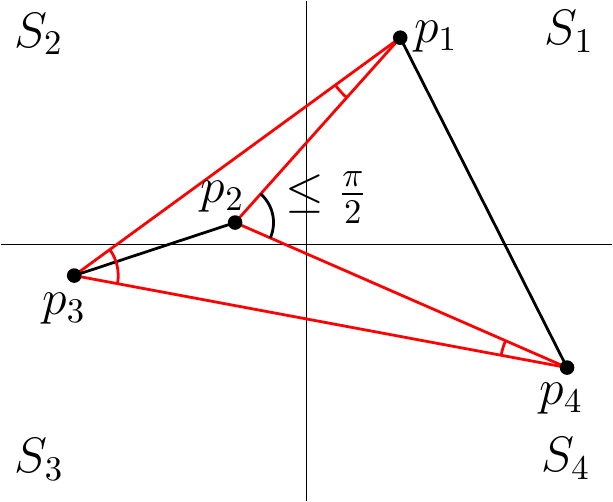}}
		&\multicolumn{1}{m{.33\columnwidth}}{\centering\vspace{0pt}\includegraphics[width=.28\columnwidth]{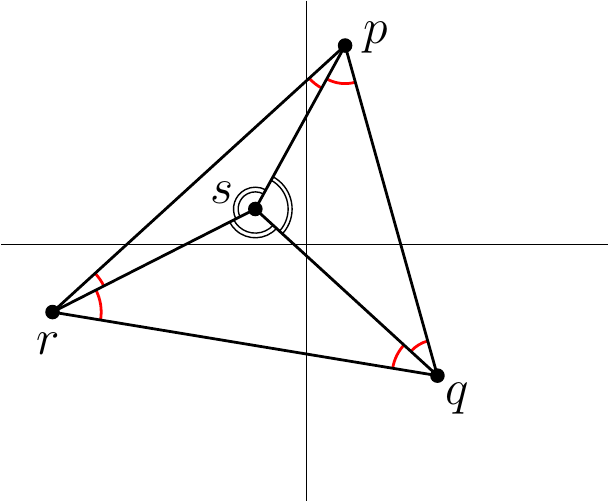}}
		\\
		(a) Convex quadruple&(b) Concave-acute quadruple &(c) Concave-obtuse quadruple
	\end{tabular}$
	\caption{Illustration of (a) Lemma~\ref{convex-acute-quadruple-lemma} where $P$ is convex and $\angle p_1xp_2\geqslant \pi/2$, (b) Lemma~\ref{convex-acute-quadruple-lemma} where $P$ is concave-acute and $\angle p_1p_2p_4\leqslant \pi/2$, and (c) Lemma~\ref{obtuse-quadruple-lemma} where all the three angles at $s$ are obtuse.}
	\label{hook-fig}
\end{figure}

Based on the above partitioning we introduce four types of quadruples. Let $P=\{p_1,p_2,p_3,p_4\}$ be a quadruple such that $p_i\in S_i$ for all $i=1,2,3,4$. We say that $P$ is {\em upward} if the path $p_2p_4p_3p_1$ (or equivalently $p_1p_3p_4p_2$) is acute, {\em downward} if the path $p_3p_1p_2p_4$ (or equivalently $p_4p_2p_1p_3$) is acute, {\em leftward} if the path $p_2p_4p_1p_3$ (or equivalently $p_3p_1p_4p_2$) is acute, and {\em rightward} if the path $p_1p_3p_2p_4$ (or equivalently $p_4p_2p_3p_1$) is acute. Such paths are referred to as ``hooks'' in \cite{Dumitrescu12}.
The following lemmas and observation, although very simple, play important roles in our proof. 

\begin{lemma}
	\label{convex-acute-quadruple-lemma}
	Let $P=\{p_1,p_2,p_3,p_4\}$ be a quadruple such that $p_i\in S_i$ for all $i=1,2,3,4$. If $P$ is convex or concave-acute then it is upward and downward or it is leftward and rightward. 
\end{lemma}
\begin{proof}
	First assume that $P$ is convex. Let $x$ denote the intersection point of the diagonals $p_1p_3$ and $p_2p_4$. If $\angle p_1xp_2 \geqslant \pi/2$ then the paths $p_2p_4p_3p_1$ and $p_3p_1p_2p_4$ are acute and thus $P$ is upward and downward; see Figure~\ref{hook-fig}(a). If $\angle p_1xp_4 > \pi/2$ then the paths $p_2p_4p_1p_3$ and $p_1p_3p_2p_4$ are acute and thus $P$ is leftward and rightward.
	
	Now assume that $P$ is concave-acute. Without loss of generality we assume that $p_2$ is the center of $P$. Observe that in this case $\angle p_1p_2p_3$ is obtuse. This and the fact that $P$ is concave-acute imply that one of $\angle p_1p_2p_4$ and $\angle p_3p_2p_4$ is acute. If $\angle p_1p_2p_4$ is acute as depicted in Figure~\ref{hook-fig}(b) then the paths $p_2p_4p_3p_1$ and $p_3p_1p_2p_4$ are acute and thus $P$ is upward and downward (observe that $\angle p_2p_1p_3 + \angle p_1p_3p_4 + \angle p_3p_4p_2 = \angle p_1p_2p_4 \leqslant \pi/2$). Analogously, if $\angle p_3p_2p_4$ is acute then the paths $p_2p_4p_1p_3$ and $p_1p_3p_2p_4$ are acute and thus $P$ is leftward and rightward.
\end{proof}


\begin{lemma}
	\label{obtuse-quadruple-lemma}
	Let $\{p,q,r,s\}$ be a concave-obtuse quadruple with center $s$. Then all angles $\angle pqs$, $\angle qps$, $\angle qrs$, $\angle rqs$, $\angle rps$, and $\angle prs$ are acute.
\end{lemma}
\begin{proof}
	See Figure~\ref{hook-fig}(c). In each of the triangles $\bigtriangleup spq$, $\bigtriangleup sqr$, and $\bigtriangleup srp$ the angle at $s$ is larger than $\pi/2$. Thus the other two angles are acute.
\end{proof}

\begin{lemma}
	\label{obtuse-quadruple-lemma2}
	Let $P=\{p_1,p_2,p_3,p_4\}$ be a quadruple such that $p_i\in S_i$ for all $i=1,2,3,4$. If $P$ is concave-obtuse then it is upward, downward, leftward, or rightward. 
\end{lemma}
\begin{proof}
	Without loss of generality assume that $p_2$ is the center of $P$. See Figure~\ref{hook-fig}(c) where $p_2=s$. In the triangle $\bigtriangleup p_1p_3p_4$ the angle at $p_1$ or the angle at $p_3$ is acute. If the angle at $p_1$ is acute then the path $p_2p_4p_1p_3$ is acute and thus $P$ is leftward ($\angle p_2p_4p_1$ is acute by Lemma~\ref{obtuse-quadruple-lemma}). If the angle at $p_3$ is acute then the path $p_2p_4p_3p_1$ is acute and thus $P$ is upward ($\angle p_2p_4p_3$ is acute by Lemma~\ref{obtuse-quadruple-lemma}).
\end{proof}

\begin{observation}
\label{opposite-quadrantobs}
Let $p$, $q$, and $r$ be any three points in $S$ such that $q$ and $r$ lie in the quadrant that is opposite to the quadrant containing $p$. Then the angle $\angle qpr$ is acute.
\end{observation}
\subsection{The tour construction}
\label{tour-section}
In this section we show how to construct an acute tour on $S$ where $|S|\geqslant 20$. By Lemma~\ref{equitable-partition-lemma} each $S_i$ with $i\in\{1,2,3,4\}$ has at least $\lfloor 20/4\rfloor=5$ points. From each $S_i$ we select an arbitrary subset of 5 points, and then we partition (the total 20) selected points into 5 quadruples such that each quadruple contains exactly one point from each $S_i$. Let $\cal Q$ denote the set of these quadruples. For any quadruple $X$ in $\cal Q$ we denote the points of $X$ by $x_1,x_2,x_3,x_4$ where $x_i\in S_i$ for all $i=1,2,3,4$. 

Since $|{\cal Q}|\geqslant 5$, by the pigeonhole principle $\cal Q$
 has three quadruples that are {\em vertical} (i.e. upward, downward, or both upward and downward) or three that are {\em horizontal} (i.e. leftward, rightward, or both leftward and rightward). Without loss of generality assume that $Q$ has three vertical quadruples. If two of these vertical quadruples are of opposite types, i.e. one upward and one downward, then we construct a tour as in case 1 below. Otherwise, the three quadruples are concave-obtuse and of the same type in which case we construct a tour as in case 2 below. Our constructions take linear time in both cases.

\begin{wrapfigure}{r}{0.45\textwidth}
	\begin{center}
		\vspace{-22pt}
		\includegraphics[width=.42\textwidth]{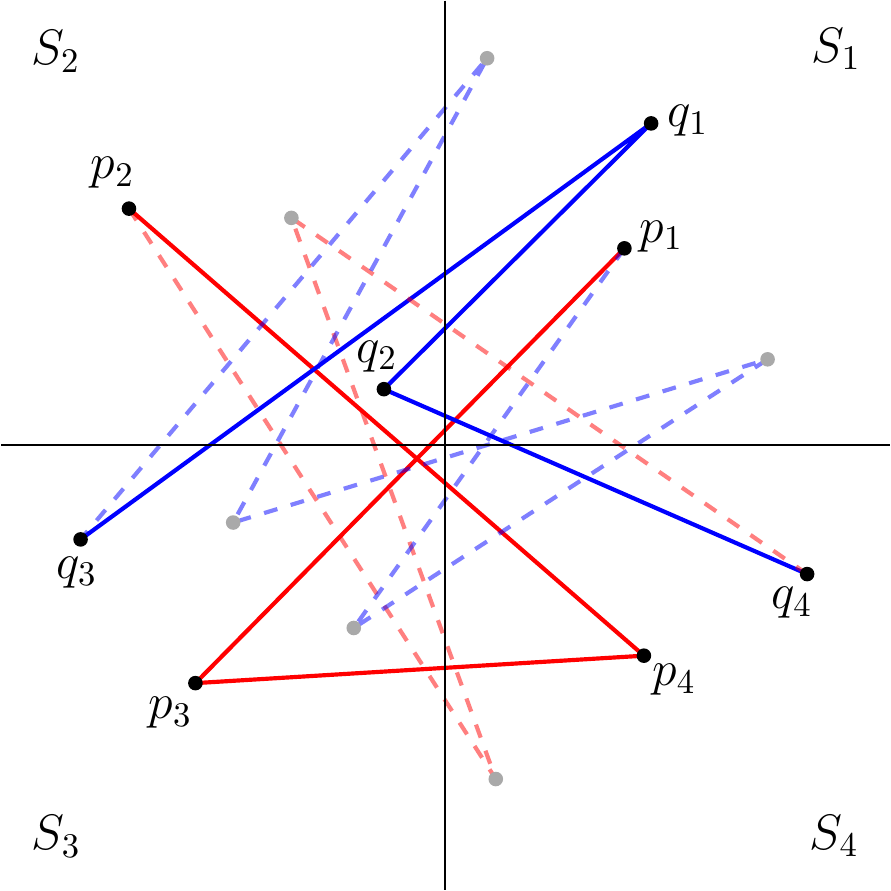}
	\end{center}
\vspace{-15pt}
\caption{Illustration of Case 1.}
\label{case-i}
	\vspace{-5pt}
\end{wrapfigure}
\paragraph{Case 1:} $\cal Q$ {\em contains two quadruples such that one is upward and the other is downward.} Let $P$ and $Q$ be such quadruples where $P$ is upward and $Q$ is downward. Since $P$ is upward, the path  $p_1p_3p_4p_2$ is acute. Since $Q$ is downward, the path $q_4q_2q_1q_3$ is acute; see Figure~\ref{case-i}. 
Let $\overline{S_2S_4}$ be a polygonal path starting from $p_2$, ending in $q_4$, alternating between $S_2$ and $S_4$, and containing all points of $S_2\cup S_4$ except for $q_2$ and $p_4$.
Let $\overline{S_3S_1}$ be a polygonal path starting from $q_3$, ending in $p_1$, alternating between $S_3$ and $S_1$, and containing all points of $S_3\cup S_1$ except for $p_3$ and $q_1$. Such polygonal paths exist because by Lemma~\ref{equitable-partition-lemma} we have $|S_2|=|S_4|$ and $|S_1|=|S_3|$. All intermediate angles of these two polygonal paths are acute by Observation~\ref{opposite-quadrantobs}. Then the tour $p_1p_3p_4p_2\oplus\overline{S_2S_4}\oplus q_4q_2q_1q_3\oplus\overline{S_3S_1}$ is acute, and it spans $S$. Notice that the angles at $p_1$, $p_2$, $q_3$ and $q_4$ are acute by Observation~\ref{opposite-quadrantobs}.

\paragraph{Case 2:}$\cal Q$ {\em contains three concave-obtuse quadruples of the same type.} Let $P$, $Q$ and $R$ be such quadruples, and without loss of generality assume that they are upward. Thus, the paths $p_2p_4p_3p_1$ and $q_2q_4q_3q_1$ and $r_2r_4r_3r_1$ are acute. Since $P$, $Q$ and $R$ are concave-obtuse their centers should lie at endpoints of these paths (the centers cannot be interior vertices of acute paths). Thus the center of $P$ is either $p_1$ or $p_2$, the center of $Q$ is either $q_1$ or $q_2$, and the center of $R$ is either $r_1$ or $r_2$. This means that the centers lie in quadrants 1 and 2. By the pigeonhole principle, and after a suitable reflection, we may assume that at least two of the centers lie in quadrant 2. After a suitable relabeling assume that the centers of $P$ and $Q$ (i.e. $p_2$ and $q_2$) lie in quadrant 2. The center of $R$ lies either in quadrant 2 (i.e. it is $r_2$) or in quadrant 1 (i.e. it is $r_1$).

After a suitable relabeling assume that $p_2$ lies below $q_2$, as in Figure~\ref{proof-fig2}. Now we build our tour as follows. First we connect $p_2$ to $p_1$ and $q_1$. The point $p_2$ is below $p_1$ because $p_2$ lies below the segment $p_1p_3$. The point $p_2$ is also below $q_1$ because $p_2$ is below $q_2$ which is in turn below $q_1$ (as $q_2$ lies below the segment $q_1q_3$). Thus $p_2$ is below both $p_1$ and $q_1$. Also notice that $p_2$ is to the left of both $p_1$ and $q_1$. Thus, the angle $\angle p_1p_2q_1$ is acute (imagine moving the origin to $p_2$, then both $p_1$ and $q_1$ would lie in the first quadrant).
Then we connect $q_3$ to $q_1$ and $q_4$. The angle $\angle q_4q_3q_1$ is acute because $Q$ is upward (i.e. the path $q_2q_4q_3q_1$ is acute). The angle $\angle p_2q_1q_3$ is acute because both $p_2$ and $q_3$ lie below and to the left of $q_1$. Therefore, the path $p_1p_2q_1q_3q_4$ is acute; see Figure~\ref{proof-fig2}.
In the rest of the construction we distinguish two subcases, depending on the center of $R$.

\begin{figure}[htb]
	\centering
	\setlength{\tabcolsep}{0in}
	$\begin{tabular}{cc}
		\multicolumn{1}{m{.5\columnwidth}}{\centering\vspace{0pt}\includegraphics[width=.45\columnwidth]{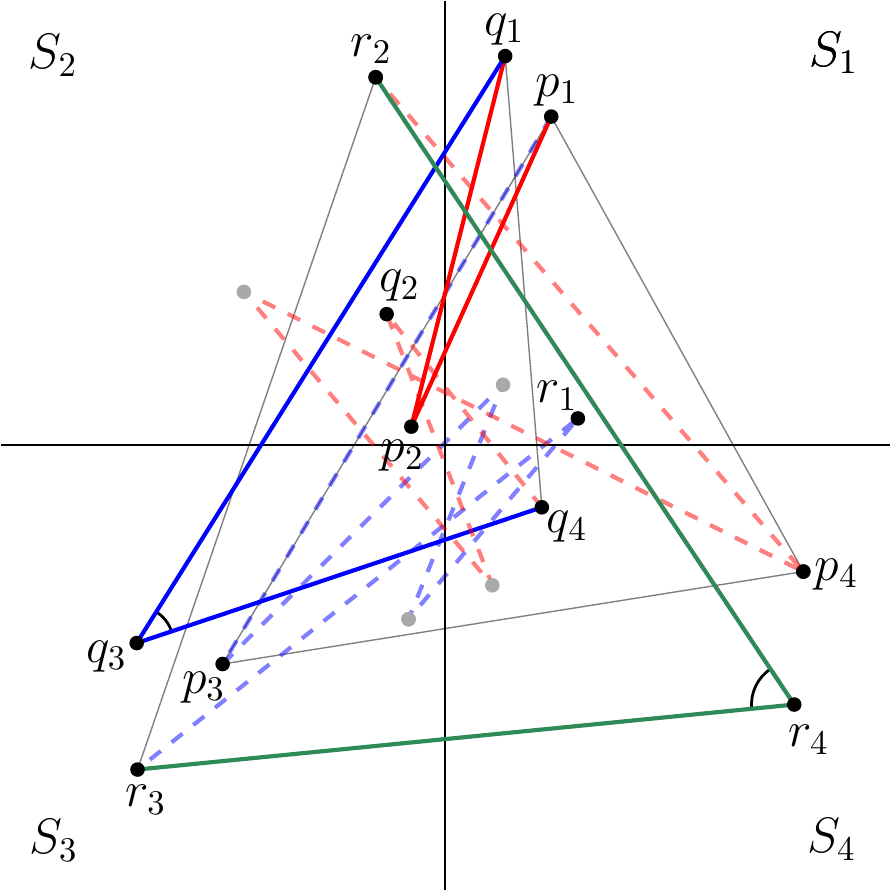}}
		&\multicolumn{1}{m{.5\columnwidth}}{\centering\vspace{0pt}\includegraphics[width=.45\columnwidth]{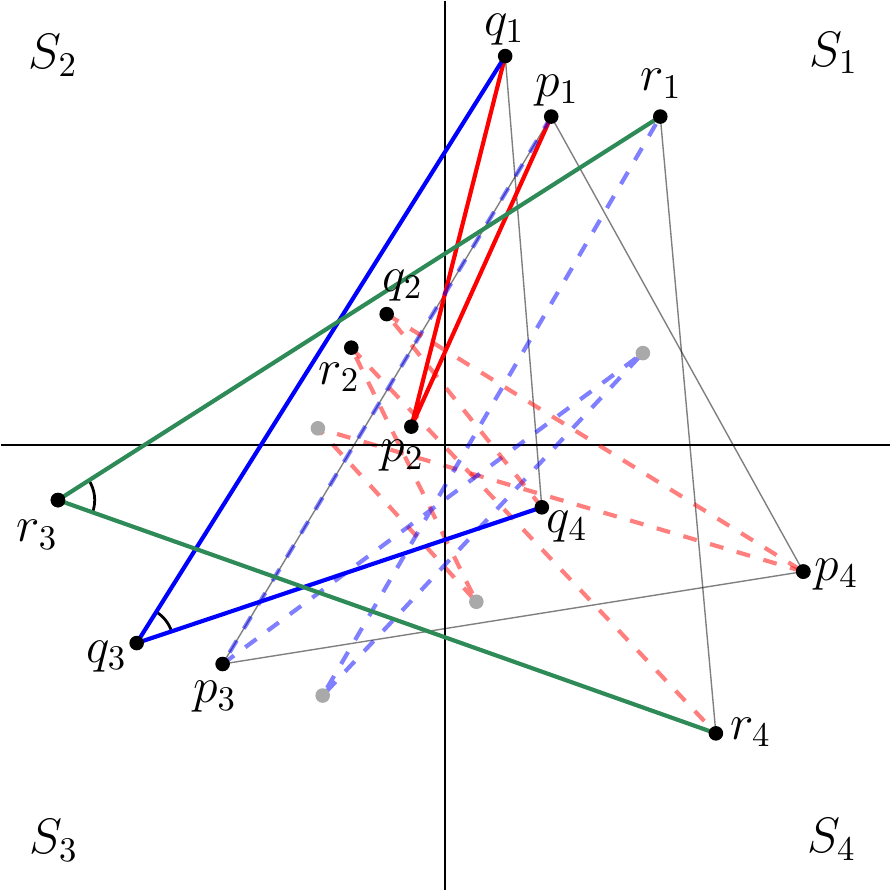}}
		\\
		(a)   &(b) 
	\end{tabular}$
	\caption{Illustration of Case 2. Three concave-obtuse quadruples $P$, $Q$ and $R$ that are upward, and the centers of $P$ and $Q$ lie in quadrant 2. (a) Subcase 2.1 where the center of $R$ is in quadrant 1. (b) Subcase 2.2 where the center of $R$ is in quadrant 2.}
	\label{proof-fig2}
\end{figure}

\paragraph{Subcase 2.1:} {\em The center of $R$ is $r_1$.} This case is depicted in Figure~\ref{proof-fig2}(a). We connect $r_4$ to $r_2$ and $r_3$. The resulting path $r_2r_4r_3$ is acute (because $R$ is upward, i.e. the path $r_2r_4r_3r_1$ is acute). 
Let $\overline{S_4S_2}$ be a polygonal path starting from $q_4$, ending in $r_2$, alternating between $S_4$ and $S_2$, containing all points of $S_4\cup S_2$ except for $r_4,p_2$, and having $q_4q_2$ as its first edge. Let $\overline{S_3S_1}$ be a polygonal path starting from $r_3$, ending in $p_1$, alternating between $S_3$ and $S_1$, containing all points of $S_3\cup S_1$ except for $q_3,q_1$, and having $r_3r_1$ as its first edge and $p_3p_1$ as its last edge. All intermediate angles of these two paths are acute by Observation~\ref{opposite-quadrantobs}. By interconnecting the constructed paths we obtain the tour  $p_1p_2q_1q_3q_4\oplus\overline{S_4S_2}\oplus r_2r_4r_3\oplus\overline{S_3S_1}$ which is acute, and it spans $S$. The angles at $p_1,r_3,q_4$ are acute by Lemma~\ref{obtuse-quadruple-lemma}, and the angle at $r_2$ is acute by Observation~\ref{opposite-quadrantobs}.

\paragraph{Subcase 2.2:} {\em The center of $R$ is $r_2$.} This case is depicted in Figure~\ref{proof-fig2}(b). We connect $r_3$ to $r_4$ and $r_1$. The resulting path $r_4r_3r_1$ is acute (because $R$ is upward, i.e. the path $r_2r_4r_3r_1$ is acute).  
Let $\overline{S_4S_2S_4}$ be a polygonal path starting from $q_4$, ending in $r_4$, alternating between $S_4$ and $S_2$, containing all points of $S_4\cup S_2$ except for $p_2$, and having $q_4q_2$ as its first edge and $r_2r_4$ as its last edge.
Let $\overline{S_1S_3S_1}$ be a polygonal path starting from $r_1$, ending in $p_1$, alternating between $S_1$ and $S_3$, containing all points of $S_1\cup S_3$ except for $q_1,q_3,r_3$, and having $p_3p_1$ as its last edge. Intermediate angles of these paths are acute by Observation~\ref{opposite-quadrantobs}.
Thus  $p_1p_2q_1q_3q_4\oplus\overline{S_4S_2S_4}\oplus r_4r_3r_1\oplus\overline{S_1S_3S_1}$ is an acute spanning tour.
The angles at $q_4$, $r_4$, and $p_1$ are acute by Lemma~\ref{obtuse-quadruple-lemma}, and the angle at $r_1$ is acute by Observation~\ref{opposite-quadrantobs}.
This finishes our proof of Theorem~\ref{main-thr}.

\section{Concluding remarks}
We showed how to construct an acute tour on any set of $n$ points in the plane, where $n$ is even and at least $20$. Our construction uses at most 12 points in each case (namely the points of quadruples $P$, $Q$ and $R$). One might be interested to extend the range of $n$ (to smaller even numbers) by taking advantage of the $8$ unused points, although this may require some case analysis. 

\paragraph{Acknowledgement.} 
I am very grateful to the anonymous SoCG 2022 reviewer who meticulously verified our proof, and provided valuable feedback that reduced the number of subcases to two (which was three in our original proof) and improved the bound on $n$ to $20$ (which was $36$ originally). 

\bibliographystyle{abbrv}
\bibliography{Acute-Tours.bib}
\end{document}